\documentclass[11pt,letterpaper]{article}

\usepackage[letterpaper,margin=1in]{geometry}

\usepackage{float}
\usepackage{array, makecell} %
\usepackage{xcolor,colortbl}
\usepackage{framed,color}
\usepackage{xcolor}
\usepackage{enumitem}
\usepackage{thm-restate}
\usepackage{natbib}

\usepackage{amsthm,amsmath,amssymb}
\usepackage[bookmarks=true,hypertexnames=false,pagebackref]{hyperref}
\hypersetup{colorlinks=true, citecolor=blue, linkcolor=red,
  urlcolor=blue}
\usepackage{newpxtext}
\usepackage[libertine,vvarbb]{newtxmath}

\usepackage[T1]{fontenc} 
\usepackage{textcomp} 

\usepackage[scr=rsfso]{mathalfa}
\usepackage{nicefrac}
\usepackage{cleveref}
\usepackage{graphicx}
\usepackage{aliascnt}
\usepackage[section]{placeins}
\usepackage{tcolorbox}
\usepackage{epigraph}
\usepackage{xifthen}
\usepackage{latexsym}
\usepackage{framed}
\usepackage{fullpage}
\usepackage{bm}
\usepackage{braket}

\newtheorem{theorem}{Theorem}[section]
\newtheorem{claim}[theorem]{Claim}

\newtheorem{lemma}[theorem]{Lemma}
\newtheorem{definition}[theorem]{Definition}

\newtheorem{remark}[theorem]{Remark}
\newtheorem{proposition}[theorem]{Proposition}

\newcommand{\polylog}{\mathrm{polylog}}

\newcommand{\supp}{\mathrm{supp}}

\newcommand{\ber}{\mathrm{Bernuolli}}

\newcommand{\E}{\mathbb{E}}
\newcommand{\adv}{\mathrm{Adv}}
\newcommand{\poi}{\mathrm{Poi}}
\newcommand{\bin}{\mathrm{Binomial}}

\title{Multi-Pass Streaming Lower Bounds for Uniformity Testing}
\author{
Qian Li\thanks{Shenzhen International  Center For Industrial  And  Applied  Mathematics, Shenzhen Research Institute of Big Data. Email: \texttt{liqian.ict@gmail.com}}
\quad
Xin Lyu\thanks{UC Berkeley. Email: \texttt{xinlyu@berkeley.edu}}
}
\date{}

\begin{document}
\maketitle

\begin{abstract}
We prove multi-pass streaming lower bounds for uniformity testing over a domain of size $2m$. The tester receives a stream of $n$ i.i.d. samples and must distinguish (i) the uniform distribution on $[2m]$ from (ii) a Paninski-style planted distribution in which, for each pair $(2i-1,2i)$, the probabilities are biased left or right by $\epsilon/2m$. We show that any $\ell$-pass streaming algorithm using space $s$ and achieving constant advantage must satisfy the tradeoff $sn\ell=\tilde{\Omega}(m/\epsilon^2)$. This extends the one-pass lower bound of Diakonikolas, Gouleakis, Kane, and Rao (2019) to multiple passes.

Our proof has two components. First, we develop a hybrid argument, inspired by Dinur (2020), that reduces streaming to two-player communication problems. This reduction relies on a new perspective on hardness: we identify the source of hardness as uncertainty in the bias directions, rather than the collision locations. Second, we prove a strong lower bound for a basic two-player communication task, in which Alice and Bob must decide whether two random sign vectors $Y^a,Y^b\in\{\pm 1\}^m$ are independent or identical, yet they cannot observe the signs directly—only noisy local views of each coordinate. Our techniques may be of independent use for other streaming problems with stochastic inputs. 

\end{abstract}

\section{Introduction}
Uniformity testing is one of the most fundamental distribution testing tasks: given $n$ independent samples from an unknown distribution $P$ over a domain of size $2m$, the goal is to distinguish whether $P$ is uniform or $\epsilon$-far from uniform under the total variation distance. Its simplicity makes it a canonical benchmark for understanding the interplay between statistical and computational resources, and insights obtained here routinely inform more complex testing problems; see \cite{survey} for a recent survey on uniformity testing and related problems. 

In this paper, we will focus on the streaming model, in which samples $X_1, \ldots, X_n$ are given as a stream and the algorithm must operate with limited memory $s$. The streaming setting captures modern large-scale learning scenarios: the learning algorithm scans the massive dataset sequentially, processing samples one at a time and updating the parameters continuously.

Uniformity testing in the streaming model has been studied in several works \citep{colt19,meir2020comparison,berg2022memory,canonne2024simpler}.
On the upper bound side, there is a folklore one-pass tester (see e.g., \cite{berg2022memory}) which uses $s=\tilde{O}(1)$ bits of memory and $n=O(m/\epsilon^2)$ samples. Here, $\tilde{O}$ and $\tilde{\Omega}$ omit $\polylog(m+n+1/\epsilon)$ terms. Beyond this small space regime, \cite{colt19} developed a streaming algorithm achieving the sample-space tradeoff $sn=\tilde{O}(m/\epsilon^4)$ when $s\leq \min\{m^{0.9},n\log m\}$. The applicable space regime was later extended to $s\leq \min\{m\log m,n\log m\}$ by \cite{canonne2024simpler}. On the lower bound side, \cite{colt19} established an unconditional lower bound $s n=\Omega(m/\epsilon^2)$, matching the folklore tester when $s=\tilde{O}(1)$; they also showed that the upper bound $sn=\tilde{O}(m/\epsilon^4)$ is tight for large memory size, specifically for $s=\tilde{\Omega}\left(m^{0.24}/\epsilon^{8/3}+m^{0.1}/\epsilon^4\right)$. 

All of the above results, however, apply only to one-pass streaming algorithms. In practice, learning algorithms typically make multiple scans over the training samples, especially when the available data is limited. This naturally motivates the study of the multi-pass setting. \cite{colt19} identified the analysis of uniform testing under multiple passes as an open problem; however, extending their information-theoretic lower bound technique to the multi-pass setting seems challenging.

In this paper, we extend the unconditional lower bound $s n=\tilde{\Omega}(m/\epsilon^2)$ to the multi-pass streaming model (Theorem \ref{thm:main0}), showing that the folklore tester remains tight even when a polylogarithmic number of passes is allowed.
 Following prior works, we employ the Paninski problems as the hard instance: given a stream $X_1,\cdots,X_n$ of i.i.d. samples from an unknown distribution $P$ over $[2m]$,
the goal is to distinguish between
\begin{itemize}
\item \textbf{Uniform distribution}: $P_i=\frac{1}{2m}$ for $i\in[2m]$. 
\item \textbf{Planted bias distribution}: Draw a \emph{bias direction} $Y=(Y_1,\cdots,Y_m)\in\{\pm 1\}^m$ uniformly at random. For each pair $(2i-1,2i)$ with $i\in[m]$, set $\left(P_{2i-1},P_{2i}\right)=\left(\frac{1+Y_i\epsilon}{2m},\frac{1-Y_i\epsilon}{2m}\right)$. 
\end{itemize}
All of the mentioned bounds above for uniformity testing also apply to the Paninski problem.
\begin{theorem}[Main Theorem]\label{thm:main0}
If a $\ell$-pass streaming algorithm solves the Paninski problem (and thus also uniformity testing), using $n$ samples and $s$ bits of memory, then $\ell s n=\tilde{\Omega}(m/\epsilon^2)$.
\end{theorem}


Our bound decreases linearly in the number of passes, suggesting that a stream of $\ell\cdot n$ fresh samples is at least as useful as $\ell$ passes over $n$ samples.

\begin{remark} The Paninski problem itself is a basic stochastic streaming problem of independent interest. For example,  
\cite{crouch2016stochastic} studied the task of estimating the collision probability $\sum_i p_i^2$ within a multiplicative error of $\epsilon$ in the streaming model. They showed that for any $n=\Omega_{\epsilon}(\sqrt{m})$, it is sufficient to have $st=\tilde{O}_{\epsilon}(m)$. Through a reduction to the needle problem, \cite{LZ23} established a multi-pass lower bound $st\ell=\tilde{\Omega}(m/\epsilon)$, demonstrating the tightness of the upper bound. Using the Paninski problems as the hard instance, we directly recover the same lower bound.
\end{remark}

\subsection{Proof Approach}
The proof of Theorem \ref{thm:main0} consists of two steps.
\vspace{-1ex}

\paragraph{Step I: hybrid argument} We develop a hybrid argument that reduces streaming to two-player communication problems. It hinges on a different view of the source of hardness: rather than in locating collisions, we understand the hardness stemming from resolving the unknown bias directions $Y\in\{0,1\}^m$. Indeed, if $Y$ were revealed, by mapping each sample to the heavy side, the Paninski problem collapses to distinguishing a fair coin from a $\epsilon$-biased one, and thus can be solved quite efficiently.

Motivated by this perspective, we define a family of hybrid distributions $\{D^{k,n/k}\}$ over streams of length $n$: partition the stream into $k$ contiguous blocks of size $n/k$, and draw each block from the planted bias distribution with a fresh bias direction $Y^i$. The planted case corresponds to $D^{1,n}$ (a single hidden direction used for the entire stream), and the extreme $D_{n,1}$ (a fresh direction per sample) is essentially the uniform case. For two distributions $D_0,D_1$ over random streams, let $\adv_{s,\ell}(D_0,D_1)$ denote the maximum distinguishing advantage achievable by an $\ell$-pass streaming algorithm using $s$ bits of space. Then we can obtain a recurrence:
\[
f(n):=\adv_{s,\ell}(D_{n,1},D_{1,n})\leq 2f(n/2)+\adv_{s,\ell}(D_{2,n/2},D_{1,n}),
\]
so it suffices to upper bound the advantage in the one-swap case $\adv_{s,\ell}(D_{2,n/2},D_{1,n})$. Equivalently, we partition the stream into two halves, each generated with a single bias direction, and ask whether the two halves use identical or independent directions.


\vspace{-1ex}

\paragraph{Step II: the hidden-sign problem} In this step, we analyze a clean two-player communication problem (Definition \ref{def:d-game}): Alice and Bob must decide whether two hidden sign vectors $Y^a,Y^b\in\{\pm 1\}^m$ are identical or independent; however neither player can see the signs directly, but only noisy local views per coordinate. In our setting, for each sign $Y_i\in\{-1,1\}$, the player can see $\poi(n/m)$ independent samples from $\ber(1/2+Y_i\epsilon/2)$.


We prove that any $C$-bit communication protocol has distinguishing advantage at most $\tilde{O}(C\epsilon^2\cdot (n/m))$. Depending on whether $n$ is greater than $m$ or otherwise, the argument is slightly different.

\medskip\noindent\underline{When $n\le m$.} Here the two sides observe (independent) fractions of the $m$ coordinates. Intuitively it is the case that both players must first search for ``shared coordinates''— indices $i$ such that both players observe samples about the $i$-th  coordinate. The hardness of this problem can be connected to the well-known problem of Set Disjointness, where players hold subsets of $[n]$ and they want to decide/search for intersections. A general phenomenon here is that on average they need to communicate $m/n$ bits to agree on a new common coordinate. Moreover, every new coordinate only offers $\epsilon^2$ distinguishing advantage: essentially, this stems from the fact that the TV distance between a pair of uniform and independent bits $(X,Y)$, and a pair of marginally-uniform but $\epsilon^2$-correlated bits $(X',Y')$, is exactly $\epsilon^2$. 

So, overall, the intuition of the lower bound can be described as: with $C$ bits of communication, the two players can find roughly $\frac{Cn}{m}$ ``common'' coordinates, on which both of them receive samples. Each of these coordinates increases their advantage by at most $\epsilon^2$, and we take a union bound to obtain the conclusion. 

Our proof formalizes the intuition through a series of reductions, starting from the lower bound for the Unique Set Disjointness problem with small advantage \cite{stoc13,random14,dinur2020streaming}.

\medskip\noindent\underline{When $n\ge m$.} This is the case where the two sides (with high probability) observe samples from most if not all coordinates. Moreover, players typically get multiple independent samples per coordinate. Intuitively this makes the distinguishing task easier: consider an extreme case, where players receive arbitrarily many samples per coordinate. Then, by taking the majority votes per coordinate, players can recover $Y^a$ and $Y^b$ locally and solve the problem with a constant communication complexity. This shows that the number of available samples per coordinate (namely $n/m$) will play a key role in the analysis.

We will prove a key reduction result, which intuitively says that for any $k$ moderately large (say, larger than $\log(m)$), $\poi(k)$ many samples from $\ber(1/2+Y_i\epsilon/2)$ is roughly as useful as a single sample from $\ber(1/2+Y_i\epsilon'/2)$ for some $\epsilon' \approx \epsilon \sqrt{k\log(1/\delta)}$, up to a statistical slackness of $\delta$. Our reduction draws on amplification/composition techniques from differential privacy (e.g.~\cite{DworkRV10}). 
In the end, we can reduce the case of $n\gg m$ to the case that $n\approx m$ but with larger $\epsilon' \approx \epsilon\sqrt{n/m}$. Then we use the already established bound of $\tilde{O}(C(\epsilon')^2(n/n)) = \tilde{O}(C\epsilon^2\cdot (n/m))$ to complete the proof.


\subsection{Related work}\label{sec:related works}
A large body of works studies streaming problems with stochastic inputs \citep{GM07,ToC16,Andoni08,CJP08,GM09,crouch2016stochastic,SODA2020randomordermatching,braverman2020coin,LZ23,stoc24,colt25}, with applications across statistical inference \citep{DBLP:conf/focs/Raz16,Sharan19,colt19,colt22} and cryptography \citep{dinur2016memory,tcc-2018-28986,DBLP:conf/eurocrypt/JaegerT19,dinur2020streaming}. Establishing space lower bounds for multi-pass streaming algorithms remains challenging. Many space lower bounds \citep{DBLP:conf/focs/Raz16,colt19,braverman2020coin,coin21,colt22} are restricted to one pass due to technical barriers, and the authors leave multi-pass bounds as major open problems. The toolbox for multi-pass lower bounds is comparatively limited. A general method reduces multi-pass streaming to communication complexity, e.g. \citep{LZ23,dinur2020streaming}. Recently, \cite{stoc24} introduced a multi-pass information complexity framework and obtained tight space lower bounds for the coin and needle problems.

\paragraph{Paper organization} Section~\ref{sec:notation} introduces the preliminaries and notations. Section~\ref{hybrid} presents the hybrid argument that reduces streaming to two-player communication problems, namely hidden-sign problems. Section \ref{sec:hiddensign} establishes the lower bound for the hidden-sign problems. Section~\ref{sec:conclusion} concludes this paper.

\section{Preliminaries}\label{sec:notation}

Generally, given a finite space $\Omega$, we use $x\sim \Omega$ to denote a random variable $x$ that is drawn uniformly at random form $\Omega$. Similarly, for a distribution $D$, we write $x\sim D$ to denote that $x$ is drawn according to $D$. Given two distributions $D_0$ and $D_1$ over $\Omega$, their \emph{total variation} (TV) distance $d_{TV}(D_0,D_1)$ is defined as $\frac{1}{2}\sum_{\omega\in\Omega}\big|\Pr[D_0=\omega]-\Pr[D_1=\omega]\big|$.

\subsection{Poissonization}

We need the technique of Poissonization. Let us review this standard technique in the below. 

\begin{definition}
    A Poisson distribution with parameter $\lambda$, denoted by $\poi(\lambda)$, is a discrete distribution over $\mathbb{N}$ with density $\Pr[\poi(\lambda) = k] = \frac{\lambda^k e^\lambda}{k!}$ for every $k\ge 0$.
\end{definition}

We rely on the following well-known fact. Let $P$ be a distribution over $[m]$. Consider the random variable $(x_1,\dots, x_m)$ sampled as follows: for each $i$ let $x_i\sim \poi(\lambda P_i)$. Also consider the random variable $(y_1,\dots, y_m)$ sampled as follows: draw $n\sim \poi(\lambda)$ and draw $z_1,\dots, z_m\sim_{i.i.d.} P$. Then let $y_i$ be the number of $i$'s among $(z_1,\dots, z_m)$. It follows that $(x_1,\dots, x_m)$ and $(y_1,\dots, y_m)$ are \emph{identically} distributed.



\subsection{Indistinguishability between distributions} Let $\epsilon,\delta\in[0,1]$ be two parameters. We say that $D_0$ and $D_1$ are $(\epsilon,\delta)$-indistinguishable, if there exist $D'_0,D'_1, D^e_0, D^e_1$ so that we can write $D_b$ as a mixture distribution $D_b = (1-\delta)D'_b + \delta D^e_b$ for both $b\in \{0,1\}$, and $D'_0$ and $D'_1$ have max-divergence bounded by $e^\epsilon$ in \emph{both} directions. More formally, for every $\omega$, it holds that $e^{-\epsilon} \Pr[D'_1 = \omega] \le \Pr[D'_0=\omega] \le e^\epsilon \Pr[D'_1 = \omega]$. Note that it follows by definition that $D'_0,D'_1$ are a pair of $(\epsilon, 0)$-indistinguishable distributions.

The following lemma in the differential privacy literature will be key to our analysis.

\begin{lemma}[See e.g.~\cite{DworkRV10}]\label{lem:dp}
    For every $\gamma \in (0,1/2)$ and $t\in \mathbb{N}$, and any desired $\delta \in [0,1]$, it holds that $\mathrm{Ber}(\frac{1}{2} - \gamma)^{\otimes t}$\footnote{Here we use $D^{\otimes t}$ to denote the distribution over $t$ independent samples from $D$.} and $\mathrm{Ber}(\frac{1}{2} - \gamma)^{\otimes t}$ are $(O(\gamma \sqrt{t\log(1/\delta)}), \delta)$-indistinguishable.
\end{lemma}

%

\subsection{Unique Set Disjointness}

\paragraph{Unique Set-Disjointness.} In the Unique Set-Disjointness (UDISJ) game, Alice and Bob each is given a subset $S,T\subseteq[3n]$ with $|S|=|T|=n$ respectively, and promised that either $|S\cap T|=1$ or $0$. Their goal is to determine $|S\cap T|$.



\begin{lemma}[\citep{stoc13,random14,dinur2020streaming}]\label{lem:udisj}
Any public-coin randomized protocol for UDISJ with advantage $\gamma$ must communicate at least $\frac{1}{20}\gamma n-20\log n$ bits in the worst case. In other words, any $C$-bit public-coin randomized protocol for UDISJ has advantage
\[
\mathrm{Adv}=\min_{(S,T)\in UDIJS^{-1}(0)}\{\Pr[\Pi(S,T)=0\}+\min_{(S,T)\in UDIJS^{-1}(1)}\{\Pr[\Pi(S,T)=1\}-1\leq \frac{20C}{n}+\frac{400\log n}{n}.
\]
\end{lemma}

\section{Hybrid Methods}\label{hybrid}


Let $D_{\mathrm{unif}}$ denote a stream of $n$ uniform samples from $[2m]$, and $D_{\mathrm{bias}}$ a stream of $n$ elements drawn from the planted bias distribution. We aim to upper bound the achievable advantage on distinguishing between $D_{\mathrm{unif}}$ and $D_{\mathrm{bias}}$.
\begin{definition}[Advantage]
 Given two distributions $D_0,D_1$ of random streams, we define 
\[
\mathrm{Adv}_{s,\ell}(D_0,D_1)=\max_{\mbox{s-space $\ell$-pass algorithm }\mathcal{A}}\left|\Pr[\mathcal{A}(D_0)=1]-\Pr[\mathcal{A}(D_1)=1]\right|
\]
We will omit the subscripts $s,\ell$ if they are clear from the context.
\end{definition}

\subsection{Hybrid Distributions}

For a sign vector $Y\in \{\pm 1\}^m$, define $P_{Y}$ as the planted bias distribution with bias direction $Y$. Equivalently, we can view $P_Y$ as a distribution over $[m]\times\{0,1\}$: to draw a sample $(i,w)$ from $P_Y$, one first draws $i$ uniformly from $[m]$; conditioning on $i$, one draws $w\sim \ber(1/2-\epsilon Y_i)$. Intuitively, $Y$ specifies a sequence of $\epsilon$-biased coins, where the $i$-th coin is biased according to $Y_i\in\{\pm 1\}$. The distribution $P_Y$ is just the uniform distribution over coins plus one flip from the selected coin.

Consider the following definition.
\begin{definition}[hybrid distributions] 
Suppose $k\in[n]$ and $k$ divides $n$. Let $D^{k,\frac{n}{k}}$ be a distribution over sequences of elements. To sample from $D^{k,\frac{n}{k}}$, we first draw $k$ sequences $X^1,\dots, X^k$ as follows. For each $i\in [k]$:
\begin{enumerate}
\item first draw a $Y^{i}\sim\{0,1\}^m$;
\item then draw $X^i$ as a sequence of $\frac{n}{k}$ independent samples from $P_{Y^i}$. 
\end{enumerate}
Finally, we define the concatenation of $X^1,\dots, X^k$ as our final sequence.
\end{definition}



It is easy to observe that $D_{\mathrm{bias}} \equiv D^{1,n}$ and $D_{\mathrm{unif}}\equiv D^{n,1}$. In the following, we analyze the distinguishing advantage between $D^{1,n}$ and $D^{n,1}$.




\subsection{The Hybrid Method}

The core of our argument is the following derivation. Using the triangle inequality, we see that:
\begin{align*}
\mathrm{Adv}(D_{\mathrm{unif}},D_{\mathrm{bias}}) &= \mathrm{Adv}\left(D^{n,1},D^{1,n}\right)\\
&\leq\mathrm{Adv}\left(D^{n,1},D^{2,\frac{n}{2}}\right)+\mathrm{Adv}\left(D^{2,\frac{n}{2}},D^{1,n}\right)\\
&\leq \mathrm{Adv}\left(D^{n,1},D^{\frac{n}{2},1}\circ D^{1,\frac{n}{2}}\right)+\mathrm{Adv}\left(D^{\frac{n}{2},1}\circ D^{1,\frac{n}{2}},D^{2,\frac{n}{2}}\right)+\mathrm{Adv}\left(D^{2,\frac{n}{2}},D^{1,n}\right)\\
&\le 2\mathrm{Adv}\left(D^{\frac{n}{2},1},D^{1,\frac{n}{2}}\right)+\mathrm{Adv}\left(D^{2,\frac{n}{2}},D^{1,n}\right)
\end{align*}
We justify the inequalities. The second line is clearly the triangle inequality. The third line is by applying the triangle inequality among the triple $(D^{n,1}\to D^{n/2,1}\to D^{1,n/2}, D^{2,n/2})$.

To see the last line, we can first write $D^{n,1}$ as $D^{n/2,1}\circ D^{n/2,1}$ and use that $\mathrm{Adv}(D^{n/2,1},D^{n/2,1}) \ge  \mathrm{Adv}(D^{n/2,1}\circ E, D^{n/2,1}\circ E)$ for any $E$. We apply similar reasoning to the term $\mathrm{Adv}\left(D^{\frac{n}{2},1}\circ D^{1,\frac{n}{2}},D^{2,\frac{n}{2}}\right)$.

\paragraph*{The recursion.} Define $f(n)=\mathrm{Adv}(D^{n,1},D^{1,n})$. Then we have the following recursion:
\[
f(n)\leq 2f(n/2)+\mathrm{Adv}\left(D^{2,\frac{n}{2}},D^{1,n}\right)
\]

Note that if we can prove $\mathrm{Adv}(D^{2,n/2},D^{1.n})\le K\cdot n$, it would follow that $f(n) \le K n \log(n)$. Thus, to prove Theorem \ref{thm:main0}, it remains to prove the following theorem. 
\begin{theorem}\label{prop:hybrid} 
We have $\mathrm{Adv}\left(D^{2,\frac{n}{2}},D^{1,n}\right)=O\left(\frac{(s\ell+\log(n))\cdot \epsilon^2}{m}\cdot n\cdot \log^2(m)\right)$, 
\end{theorem}
In order to prove \Cref{prop:hybrid}, it suffices to consider its corresponding two-player communication problem, called Hidden Sign Problem (HSP).

\begin{definition}[Hidden Sign Problem, I]\label{def:d-game}
Let $\epsilon\in[0,1]$ and $n\in \mathbb{N}$. Consider the following two cases.
\begin{itemize}
\item Case ($\perp$): Draw $Y^a,Y^b\sim\{\pm 1\}^m$ independently. 
\item Case (=): Draw $Y^a=Y^b\sim\{\pm 1\}^m$. 
\end{itemize}
Once $Y^a,Y^b$ are drawn, Alice (resp.~Bob) receives a sequence of $n/2$ samples from $P_{Y^a}$ (resp.~$P_{Y^b}$). Their goal is to distinguish between the two cases.
\end{definition}

We will apply a Poissonization trick and analyze the following variant of the Hidden Sign Problem.

\begin{definition}[Hidden Sign Problem, II]\label{def:d-game}
Let $\epsilon\in[0,1]$ and $n\in \mathbb{N}$. Consider the following two cases.
\begin{itemize}
\item Case ($\perp$): Draw $Y^a,Y^b\sim\{\pm 1\}^m$ independently. 
\item Case (=): Draw $Y^a=Y^b\sim\{\pm 1\}^m$. 
\end{itemize}
Once $Y^a,Y^b$ are drawn, Alice and Bob receive inputs sampled as follows: 
\begin{itemize}
\item For each $i\in[m]$, Alice draw $\mathrm{Poi}(n/m)$ independent samples $\sim \ber(1/2+Y^a_i\epsilon/2)$;
\item For each $i\in[m]$, Bob draw $\mathrm{Poi}(n/m)$ independent samples $\sim \ber(1/2+Y^b_i\epsilon/2)$;
\end{itemize}
Their goal is to distinguish between the above two cases.
\end{definition}

\paragraph*{A reduction.} An equivalent way to state the input distributions for Version II of the Hidden Sign Problem is that each player will first draw an integer $n_a$ (resp.~$n_b$) from $\mathrm{Poi}(n)$, and then draw $n_a$ ($n_b$) samples from $P_{Y^a}$ (resp.~$P_{Y^b}$). It is a standard fact that $\Pr[\mathrm{Poi}(n)\ge \frac{n}{2}] \ge \frac{1}{2}$. As a consequence, with probability at least $\frac{1}{4}$, we have $\min(n_a, n_b)\ge n$. Thus, if we had a $C$-bit communication protocol that solves Version I of HSP with advantage $\xi$, we could then design a $(C+2)$-bit protocol that solves Version II with advantage at least $\frac{\xi}{4}$ (both players use $2$ bits to check that $\min(n_a,n_b)\ge n$, and run the assumed protocol with the first $n$ samples if the check passes). 

Given this reduction, the rest of the paper will focus on proving a communication lower bound for the Hidden Sign Problem (Version II). That is, we will prove the following theorem.

\begin{theorem}\label{thm:hidden-sign-lb}
For every $\epsilon, m$, $n\leq O(m/\epsilon^2)$, and $C\geq \log m$, it holds that that any $C$-bit protocol for the Hidden Sign Problem (II) achieves a distinguishing advantage of at most $O\left(\log^2(m)\cdot \frac{n}{m} \cdot C\cdot \epsilon^2 \right)$.
\end{theorem}

We note that \Cref{thm:hidden-sign-lb} implies \Cref{prop:hybrid}, and hence the main result of our paper.

\section{Proof of \Cref{thm:hidden-sign-lb}}\label{sec:hiddensign}

In this section, we prove \Cref{thm:hidden-sign-lb}. This is achieved via a series of reductions.

\subsection{The Hidden Index Problem}

\paragraph*{Introducing HIP.} As the first step, let us consider the following two-player communication complexity game, called Hidden Index Problem (Definition \ref{def:hip}), and prove its lower bound (Theorem \ref{thm:hip}).

\begin{definition}[Hidden Index Problem, $\mathrm{HIP}_{n}$]\label{def:hip}
Alice and Bob each is given a randomized string $a,b\in\{1,0,\star\}^{3n}$ sampled in the following way. First, $\supp(a)$ and $\supp(b)$ are sampled as follows\footnote{We define $\supp(a)$ as $\{i:a_i\neq \star\}$.}.
\begin{itemize}
\item Draw $i\sim[3n]$, and two \emph{disjoint} random subsets $S_1,S_2$ each of size $n-1$ from $[3n]\setminus\{i\}$.
\item $\mathrm{supp}(a)=S_1\sqcup\{i\}$ and $\mathrm{supp}(a)=S_2\sqcup\{i\}$.
\end{itemize}
Note that $|a|=|b|=n$ and $|a\cap b|=1$. Their goal is to distinguish the following two cases: 
\begin{itemize}
\item $D_=$:  $a_i=b_i\sim\{0,1\}$; 
\item $D_\neq$:  $a_i=1-b_i\sim\{0,1\}$.
\end{itemize}
For any other $a_j$ or $b_k$ in the support, it is independently uniformly drawn from $\{0,1\}$. Their goal is to distinguish between the above two cases.
\end{definition}

For a randomized communication protocol $\Pi$ , its advantage for $\mathrm{HIP}_n$ is defined as
\[
\adv_{n}^{HIP}(\Pi)=\Pr_{D_=}[\Pi(a,b)\mbox{ outputs} \text{``$=$''}]+\Pr_{D_\neq}[\Pi(a,b)\mbox{ outputs} \text{``$\ne$''}]-1.
\]

The following proposition asserts that we can assume that the algorithm performs at least as good as random guesses.

\begin{proposition}\label{prop:hip}
If there exists a $C$-bit protocol $\Pi$ for HIP with advantage $\gamma\geq 0$, then there exists a related $C$-bit protocol $\Pi'$ such that both $\Pr_{D_=}[\Pi(a,b)\mbox{ outputs} \text{``$=$''}]$ and $\Pr_{D_\neq}[\Pi(a,b)\mbox{ outputs} \text{``$\ne$''}]$ are $\geq 1/2+\gamma/10$.
\end{proposition}
\begin{proof}
Let $p_1$ and $p_2$ abbreviate $\Pr_{D_=}[\Pi(a,b)\mbox{ outputs} \text{``$=$''}]$ and $\Pr_{D_\neq}[\Pi(a,b)\mbox{ outputs} \text{``$\ne$''}]$ respectively. We have $p_1+p_2=1+\gamma$, and we assume $p_1>p_2$ with loss of generality. The new protocol $\Pi'$ is constructed as follows. Let $\eta:=\frac{p_1-p_2}{1+p_1-p_2}$. Then w.p. $(1-\eta)$, it runs $\Pi$; otherwise, it outputs $\text{``$\ne$''}$ directly. Noting that $p_1'=p_2'=(1-\eta)p_1=\frac{p_1}{1+p_1-p_2}=\frac{p_1}{2p_1-\gamma}\geq \frac{1}{2-\gamma}\geq 1/2+\gamma/10$, we finish the proof.
\end{proof}

Next, we perform a worst-case to average-case reduction, to lift the lower bound of UDISJ to a lower bound for the HIP problem.

\begin{lemma}\label{thm:hip}
$\adv_{n}^{HIP}(\Pi)<\left(\frac{400C}{n}+\frac{160000\log n}{n}\right)$ for any $C$-bit public-coin randomized protocol $\Pi$.
\end{lemma}
\begin{proof}
By contradiction, assume that for some $C$, there exists a $C$-bit public-coin randomized protocol $\Pi$  for $\mathrm{HIP}_{n}$ has advantage at least $\frac{400C}{n}+\frac{160000\log n}{n}:=\gamma$. By Proposition \ref{prop:hip}, we can assume 
\[
\Pr_{D_=}[\Pi(a,b)\mbox{ outputs} \text{``$=$''}]\geq \frac{1}{2}+\frac{\gamma}{10}, \mbox{ and } \Pr_{D_\neq}[\Pi(a,b)\mbox{ outputs} \text{``$\ne$''}]\geq \frac{1}{2}+\frac{\gamma}{10}.
\]

Let $D_{out}$ be a distribution on $(a,b)\in\{0,1,\star\}^{3n}$ with $|a|=|b|=n$ and $|a\cap b|=0$ defined as follows:
\begin{itemize}
\item $\mathrm{supp}(a)$ and $\mathrm{supp(b)}$ are two disjoint random subsets from $[3n]$ each of size $n$. Each $a_j$ and $b_k$ in the support is an independently uniformly random bit.
\end{itemize}
In the following, we first focus on the case $\Pr_{D_{out}}[\Pi(a,b) \mbox{ outputs }\text{``$=$''}]\geq 1/2$. The other case $\Pr_{D_{out}}[\Pi(a,b) \mbox{ outputs }\text{``$\ne$''}]\geq 1/2$ can be handled similarly, and will be specified later.

Now, we construct a $C$-bit public-coin randomized protocol $\Pi'$ for UDISJ. Suppose Alice holds $S\subseteq[3n]$ with $|S|=n$, and Bob holds $T\subseteq[3n]$ with $|T|=n$, then the protocol proceeds as follows:
\begin{itemize}
\item Public randomness: a random string $x\in\{0,1\}^{3n}$, and a random permutation $\sigma:[3n]\rightarrow[3n]$.
\item Alice generates a $a\in\{0,1,\star\}^{3n}$ with $\mathrm{supp}(a)=\sigma(S)$ and $a_i=x_i$ for $i\in\mathrm{supp}(a)$.
\item Bob generates a $b\in\{0,1,\star\}^{3n}$ with $\mathrm{supp}(b)=\sigma(T)$ and $b_i=1-x_i$ for $i\in\mathrm{supp}(b)$.
\item Alice and Bob run $\Pi$ on $(a,b)$, and obtain $\mathrm{ans}\in\{=,\neq\}$.
\item If $\mathrm{ans}$ is $=$, outputs $0$. Otherwise, outputs $1$.
\end{itemize}
We claim that the advantage of $\Pi'$ is 
\[
\adv=\min_{(S,T)\in UDIJS^{-1}(0)}\{\Pr[\Pi'(S,T)=0\}+\min_{(S,T)\in UDIJS^{-1}(1)}\{\Pr[\Pi'(S,T)=1\}-1> \frac{20C}{n}+\frac{400\log n}{n},
\]
and reaches a contradiction with Lemma \ref{lem:udisj}. This is because
\begin{itemize}
\item If $|S\cap T|=0$, then $(a,b)\sim D_{out}$, and $\Pr[\Pi'(S,T)=0]=\Pr_{D_{out}}[\Pi(a,b)= \text{``$=$''}]\geq 1/2$.
\item If $|S\cap T|=1$, then $(a,b)\sim D_{\neq}$, and $\Pr[\Pi'(S,T)=1]=\Pr_{D_\neq}[\Pi(a,b)= \text{``$\ne$''}]\geq \frac{1}{2}+\frac{\gamma}{10}=\frac{1}{2}+\frac{40C}{n}+\frac{16000\log n}{n}$.
\end{itemize}

For the other case, $\Pr_{D_{out}}[\Pi(a,b) \mbox{ outputs }\neq]\geq 1/2$, we can handle similarly by changing  $b_i=1-x_i$ to $b_i=x_i$, and changing the reduction step to be "If $\mathrm{ans}$ is $=$, outputs $1$. Otherwise, outputs $0$."
\end{proof}

\paragraph*{The Noisy HIP.} In order to facilitate the analysis of the Hidden Sign Problem, we introduce the following \emph{noisy} Hidden Index Problem, NHIP for short.


\begin{definition}[Noisy Hidden Index Problem, $\mathrm{NHIP}_{n}$]\label{def:nhip}
Alice and Bob each is given a randomized string $a,b\in\{1,0,\star\}^{3n}$. First $\supp(a)$ and $\supp(b)$ are sampled as follows:
\begin{itemize}
\item Draw $i\sim[3n]$, and two \emph{disjoint} random subsets $S_1,S_2$ each of size $n-1$ from $[3n]\setminus\{i\}$.
\item $\mathrm{supp}(a)=S_1\sqcup\{i\}$ and $\mathrm{supp}(a)=S_2\sqcup\{i\}$.
\end{itemize}
Note that $|a|=|b|=n$ and $|a\cap b|=1$. Their goal is to distinguish the following two cases: 
\begin{itemize}
\item Case ``$=$'':  Draw $Y^a_i=Y^b_i\sim\{\pm 1\}$. For any other $a_j$ or $b_k$ in the support of $a$ and $b$, draw $Y^a_j$ ($Y^b_k$) independently from $\{\pm 1\}$.
\item Case ``$\ne$'': Draw $Y^a_i=1 - {Y^b_i}\sim\{\pm 1\}$. For any other $a_j$ or $b_k$ in the supports of $a$ and $b$, draw $Y^a_j$ ($Y^b_k$) independently from $\{\pm 1\}$.
\end{itemize}
From the hidden $Y^a_j$, Alice observes an $a_j\sim \ber(\frac{1}{2} -\epsilon Y_j^a)$ for every $j\in \mathrm{supp}(a)$. Symmetrically, from the hidden $Y^b_j$, Bob observes a $b_j\sim \ber(\frac{1}{2} -\epsilon Y_j^b)$ for every $j\in \mathrm{supp}(b)$. We emphasize that Alice and Bob cannot observe $Y^a$ ($Y^b$) directly.
\end{definition}

For NHIP, we prove a lower bound that is, roughly speaking, stronger by a factor of $\frac{1}{\varepsilon^2}$ than that for HIP.
\begin{lemma}\label{lemma:nhip}
$\adv_{n,\epsilon}^{NHIP}(\Pi)\le 2\epsilon^2\left(\frac{400C}{n}+\frac{160000\log n}{n}\right)$ for any $C$-bit public-coin randomized protocol $\Pi$.
\end{lemma}
\begin{proof}
We perform a reduction. Let $\Pi$ be a $C$-bit public-coin protocol $\Pi$ for NHIP. 

Let $D_=$ be the joint (between Alice and Bob) input distribution of NHIP for the ``$=$'' case, and let $D_{\ne}$ be that for the ``$\ne$'' case. Also, let $E_=$ (resp.~$E_{\ne}$) be the input distribution of HIP (\Cref{def:hip}) for the ``$=$'' (resp.~``$\ne$'') case.

Note that $D_{=}$ can be written as a mixture of $D^i_{=}$'s where $D^i_{\ne}$ denotes the input distribution conditioned on the event that $a\cap b = \{i\}$. Similar decompositions exist for $D_{\ne}, E_=, E_{\ne}$.  Fix one such $i$. Let us study $D_{=}^i$ and $D_{\ne}^i$ closely. One can easily see that
\begin{align*}
D_{=}^i = \left( \left(\frac{1}{2}+\epsilon\right)^2 + \left(\frac{1}{2}-\epsilon\right)^2 \right) E_=^i + 2 \left(\frac{1}{2}-\epsilon\right) \left(\frac{1}{2} + \epsilon\right) E_{\ne}^i.
\end{align*}
To see this, simply note that the inputs are induced from the same value of $Y_i^a$ and $Y_i^b$. Through the Bernoulli sampling with bias $\frac{1}{2} \pm \epsilon$, it follows that the both inputs agree on the $i$-th coordinate with probability $(\frac{1}{2}-\epsilon)^2 + (\frac{1}{2}-\epsilon)^2$, and, if they agree, they agree on either $0$ or $1$ with equal probability. Similar reasoning applies for the other case of non-agreement on the $i$-th coordinate. 

Using a similar reasoning, we also obtain that
\begin{align*}
D_{\neq}^i = \left( \left(\frac{1}{2}+\epsilon\right)^2 + \left(\frac{1}{2}-\epsilon\right)^2 \right) E_{\ne}^i + 2 \left(\frac{1}{2}-\epsilon\right) \left(\frac{1}{2} + \epsilon\right) E_{=}^i.
\end{align*}

Let us now define $E^i_{\text{common}} = \frac{1}{2} E^i_{\ne} + \frac{1}{2} E^i_{=}$, and write
\begin{align*}
    D_{=}^i &= (1-2\epsilon^2) E^i_{\text{common}} + 2\epsilon^2 E^i_{=}, \\
    D_{\ne}^i &= (1-2\epsilon^2) E^i_{\text{common}} + 2\epsilon^2 E^i_{\ne}.
\end{align*}
Define $E_{\text{common}}$ to be the uniform mixture of $E^i_{\text{common}}$. We see that
\begin{align*}
    D_= &= (1-2\epsilon^2) E_{\text{common}} + 2\epsilon^2 E_{=}, \\
    D_\ne &= (1-2\epsilon^2) E_{\text{common}} + 2\epsilon^2 E_{\ne}.
\end{align*}

We now make use of \Cref{thm:hip} and it immediately follows that
\begin{align*}
\adv_n^{NHIP}(\Pi) = 2\epsilon^2 \adv_n^{HIP}(\Pi) = 2\epsilon^2 \left( \frac{400C}{n} + \frac{160000 \log n}{n}\right),
\end{align*}
as claimed.
\end{proof}

\paragraph*{Noisy HIP with general observations.} Finally, in order to analyze the game from \Cref{def:d-game}, we will have to work with a variant of Noisy HIP where we observe not bits but some other forms of signal emitted from two similar sources. 


\begin{definition}[Noisy Hidden Index Problem with General Observation, NHIPG]\label{def:nhip-general}
Let $\epsilon,\delta$ be parameters, and $\Gamma_0,\Gamma_1$ be a pair of $(\epsilon,\delta)$-indistinguishable distributions.

We consider a similar setup as in \Cref{def:nhip}: Alice and Bob are given uniformly random subsets $A$ and $B$ of $[3n]$, each of size $n$, subject to the condition that $|A\cap B| = 1$. Then, the implicit variables $Y^a\in \{\pm 1\}^{A}$ and $Y^b\in \{\pm 1\}^B$ are drawn according to the ``$=$'' or the ``$\ne$'' case. Then, Alice and Bob each observe the following:
\begin{itemize}
    \item From Alice's side, for each $j\in A$, Alice observes $a_j\sim \Gamma_{Y^a_j}$.
    \item From Bob's side, for each $k\in B$, Bob observes $b_k\sim \Gamma_{Y^b_k}$.
\end{itemize}
The goal for Alice and Bob is to distinguish between the case of ``$=$'' and that of ``$\ne$''.
\end{definition}

Generally, the instantiation of NHIPG will depend on the choice of $\Gamma_0,\Gamma_1$. However, we will still use the notation $\adv_{n,\epsilon,\delta}^{NHIPG}(\Pi)$ to denote the \emph{maximum} achievable advantage of the protocol $\Pi$ on solving \emph{any} instantiation of NHIPG with a pair of $(\epsilon,\delta)$-indistinguishable source distributions $\Pi_0,\Pi_1$. With this notational convention, we state and prove the following lemma.

\begin{lemma}\label{lem:nhip-general}
Assume $\epsilon < 0.5$. We have $\adv_{n,\epsilon,\delta}^{NHIPG}(\Pi) \le \epsilon^2\left(\frac{400C}{n} + \frac{160000\log n}{n}\right) + 2\delta n$ for any $C$-bit public-coin protocol $\Pi$.
\end{lemma}

\begin{proof}
    An easy fact is that $\adv_{n,\epsilon,\delta}^{NHIPG}(\Pi) \le \adv_{n,\epsilon,0}^{NHIPG}(\Pi) + 2\delta n$. This follows for a simple reason: suppose the observations are drawn from $\Gamma_0,\Gamma_1$. Then, whenever Alice (resp.~Bob) is to observe a sample from $\Gamma_b$, we can think of it as observing a sample from the mixture distribution $(1-\delta)\Gamma'_b + \delta \Gamma^e_b$, and it follows that the sample is drawn from $\Gamma^e_b$ with probability $\delta$. Thus, with probability $1-2\delta n$, none of the samples given to Alice/Bob is drawn from $\Gamma^e_b$. We can therefore analyze the advantage for the pair of $(\epsilon,0)$-indistinguishable source distributions $\Gamma'_0,\Gamma'_1$.

    In the following, we will just assume $\Gamma_0,\Gamma_1$ are $(\epsilon,0)$-indistinguishable and prove an upper bound on $\adv_{n,\epsilon,0}^{NHIPG}(\Pi)$. We will make use of a well-known fact: if $\Gamma_0,\Gamma_1$ are $(\epsilon,0)$-indistinguishable, then there exists a pair of distributions $\Gamma^c_0,\Gamma^c_1$ such that we can write $\Gamma_0,\Gamma_1$ as the following mixture distributions:
    \begin{align*}
        \Gamma_0 &= \frac{e^\epsilon}{1 + e^\epsilon} \Gamma^c_0 + \frac{1}{1 + e^\epsilon}\Gamma^c_1, \\
        \Gamma_1 &= \frac{e^\epsilon}{1 + e^\epsilon} \Gamma^c_1 + \frac{1}{1 + e^\epsilon}\Gamma^c_0.
    \end{align*}

    We now describe a reduction from NHIP to NHIPG. Suppose now we are tasked to solve NHIP with noise level $\epsilon' = \frac{1}{2} - \frac{1}{1 + e^\epsilon} \le \epsilon$. We design a protocol $\Pi'$ for this using $\Pi$ as a black box. Let $(a,b)$ be the (randomized) inputs to Alice and Bob in the NHIP problem. We know that, for each $j\in \mathrm{supp}(a)$, it holds that $a_j \sim \ber(\frac{1}{2}-\epsilon' Y^a_j)$. Let us post-process $a_j$ by drawing $a'_j \sim \Gamma^c_{a_j}$ using independent random coins available to us (by the public-randomness assumption). Composing this sampling step with $a_j\sim \ber(\frac{1}{2}-\epsilon' Y^a_j)$, we see that conditioning on $Y_a^j$, $a'_j$ is distributed as
    \begin{align*}
        a'_j = (\frac{1}{2} - \epsilon') \Gamma^c_{Y_j^a} + (\frac{1}{2} + \epsilon')\Gamma^c_{1 - Y_j^a} = \Gamma_{Y_j^a}.
    \end{align*}
    Similar reasoning holds for the Bob's side as well. Hence, it follows that if we independently post-process every $a_j$ and $b_k$ for $j\in \mathrm{supp}(a)$ and $k\in \mathrm{supp}(b)$, we obtain an instance $(a',b')$ for the NHIPG problem. We then run the protocol $\Pi$ between Alice and Bob on the instance.
    
    To analyze the advantage, note that if $(a,b)$ was drawn from the ``$=$'' (resp.~``$\ne$'') case of NHIP, $(a',b')$ is distributed as the ``$=$'' (resp.~``$\ne$'') case of NHIPG. Hence, it immediately follows that
    \begin{align*}
        \adv_{n,\epsilon,0}^{NHIPG}(\Pi) = \adv_{n,\epsilon'}^{NHIP}(\Pi') \le \epsilon^2\left( \frac{400 C}{n} + \frac{160000\log n}{n} \right),
    \end{align*}
    as desired.
\end{proof}

\subsection{From Unique Intersection to Statistical Subsamples}

In the last section, we were always analyzing the case where Alice and Bob each observe a subset $[3n]$ with the promise that their observations intersect at exactly one coordinate. There is an evident gap from this to the Hidden Sign Problem: namely in the Hidden Sign Problem, each party observes a random subset with expected size $\theta n$, with \emph{no} promise on the intersection size between the two observed subsets from two parties.

Generally speaking, all the variants of hidden index problems from the last section can be described by a pair of source distributions $\Gamma_0,\Gamma_1$, from which Alice and Bob draw observations. With this in mind, we consider the following generalization of Hidden Sign Problem as follows:
\begin{definition}[Generalized Hidden Sign, GHS]\label{def:d-game-general}
Let $\Gamma_0,\Gamma_1$ be a pair of distributions. Let $n\in \mathbb{N}, \theta \in [0,1]$ be parameters. Define the statistical HiddenSign problem: Alice and Bob first sample subsets $\mathrm{supp}(a)$ and $\mathrm{supp}(b)$ via Poisson subsampling: namely each $i\in[n]$ is independently included in $\mathrm{supp}(a)$ with probability $\theta$. Similarly, each $j$ is included in $\mathrm{supp}(b)$ with probability $\theta$.

Let $Y^a,Y^b\in \{\pm 1\}^n$ be two strings of length $n$, sampled dependent on the case as follows:
\begin{itemize}
    \item Case ``$=$'': $Y^a = Y^b\sim \{\pm 1\}^n$. 
    \item Case ``$\perp$'': $Y^a$ and $Y^b$ are independently drawn from $\{\pm 1\}^n$.
\end{itemize}
Then, Alice and Bob receive inputs sampled from the following distribution:
\begin{itemize}
    \item Alice receives, for each $i\in \mathrm{supp}(a)$, a sample $a_i\sim \Gamma_{Y^a_i}$.
    \item Bob receives, for each $i\in \mathrm{supp}(b)$, a sample $b_i\sim \Gamma_{Y^b_i}$.
\end{itemize}
The goal for Alice and Bob is to distinguish between the two cases.
\end{definition}

We are ready to state the meta-theorem, as described below.
\begin{theorem}\label{thm:general-hidden-sign}
    Suppose $\Gamma_0$ and $\Gamma_1$ are $(\epsilon,\delta)$-indistinguishable. Consider the General Hidden Sign problem with $\Gamma_0,\Gamma_1$ and parameters $n,\theta$. Then, for every $C\ge \log(n)$, it holds that
    \begin{align*}
        \mathrm{Adv^{GHS}_{\theta,\epsilon,\delta}}(\Pi)\le O(\epsilon^2 C \theta + 2n\delta).
    \end{align*}
\end{theorem}

\begin{proof}
First of all, we can ``pay'' the price of $2n\delta$ and turn to analyze a pair of sources $(\Gamma_0,\Gamma_1)$ that is $(\epsilon,0)$-indistinguishable. 

\paragraph*{Proof for $\theta \in [0,1/2)$.} We first prove the bound for the case of $\theta < \frac{1}{2}$.

We use $q_{\cap}$, $q_a$, and $q_b$ to denote the sizes of $|\supp(a)\cap \supp(b)|$, $|\supp(a)|$, and $|\supp(b)|$ respectively. We use $\adv_{q_\cap,q_a,q_b}$ to denote the maximum advantage of $\Pi$ conditioned on $(q_\cap,q_a,q_b)$; that is,
\begin{align*}
\adv_{q_\cap,q_a,q_b}=\sup_{\Pi} \left|\Pr_{D_=}[\Pi(a,b)\mbox{ outputs } \text{``$=$''} \mid q_\cap,q_a,q_b]-\Pr_{D_\perp}[\Pi(a,b) \mbox{ outputs } \text{``$=$''} \mid q_\cap,q_a,q_b]\right|.
\end{align*}
We have
\[
\adv^{GHS}_{\theta,\epsilon,\delta} \le  \E\left[\adv_{q_\cap,q_a,q_b}\right].
\]
One way to see this is to note that we can grant Alice and Bob the knowledge of $q_{\cap}, q_a, q_b$, which only make the distinguishing task easier. Moreover, conditioned on $q_{\cap},q_a,q_b$, we know that $\supp(a)$ and $\supp(b)$ are two uniform subsets of size $q_a$ and $q_b$ with an intersection size $q_\cap$. 

Furthermore, let $\adv^{k}_{q_\cap,q_a,q_b}$ (for $0\leq k\leq q_\cap$) denote the maximum advantage of any protocol $\Pi$ on distinguishing the following two cases:
\begin{itemize}
\item $D_{=^k\perp^{q_\cap-k}}$: $Y^a_i=Y^b_i$ for $k$ random indices $i$ in $\supp(a)\cap \supp(b)$, and all other $Y_a^j$ and $Y^j_b$ are independently drawn.
\item  $D_{=^{k+1}\perp^{q_\cap-k-1}}$: similarly $Y^a_i=Y^b_i$ for $(k+1)$ random indices in $\supp(a)\cap \supp(b)$, and other bits of $Y_a$ and $Y_b$ are independently drawn.
\end{itemize} 
Formally, we define
\[
\adv^{k}_{q_\cap,q_a,q_b}:=\sup_{\Pi} \left|\Pr_{D_{=^{k+1}\perp^{q_\cap-k-1}}}[\Pi(a,b)\mbox{ outputs }=\mid q_\cap,q_a,q_b]-\Pr_{D_{=^k\perp^{q_\cap-k}}}[\Pi(a,b) \mbox{ outputs } =\mid q_\cap,q_a,q_b]\right|.
\]
The following claim will be central to us.
\begin{claim} It holds that
\begin{align*}
\adv_{q_\cap,q_a,q_b}\leq\sum_{k=0}^{q_\cap} \adv^k_{q_\cap,q_a,q_b}\leq q_\cap\cdot \adv_{1,q_a-q_\cap+1,q_b-q_\cap+1}\leq q_\cap\cdot \epsilon^2 C\cdot O\left(\frac{1}{q_a-q_{\cap}+1}+\frac{1}{q_b-q_\cap+1}\right).
\end{align*}
\end{claim}
\begin{proof}
The first inequality is by the triangle inequality. For the second inequality, we prove that $\adv^k_{n,q_\cap,q_a,q_b}\leq \adv_{n-q_\cap+1,1,q_a-q_\cap+1,q_b-q_\cap+1}$. Indeed, consider the natural coupling between $D_{=^k\perp^{q_\cap-k}}$ and $D_{=^{k+1}\perp^{q_\cap-k-1}}$. Namely, we couple the realizations of $D_{=^k\perp^{q_\cap-k}}$ and $D_{=^{k+1}\perp^{q_\cap-k-1}}$ in a way that, for both cases there are $k$ common indices $i\in \mathrm{supp}(a)\cap \supp(b)$ such that $Y_i^a = Y_i^b$, and there is exactly one additional $i'\in \supp(a)\cap\supp(b)$ such that $Y_{i'}^a=Y_{i'}^b$ for the case of $D_{=^{k+1}\perp^{q_\cap-k-1}}$.

Now, we can grant Alice and Bob the knowledge of the $k$ common indices for which $Y^a_i=Y^b_i$. This does not make the distinguishing game harder. Consequently,
\begin{align*}
    \adv^k_{q_\cap,q_a,q_b}\leq \adv_{1,q_a-q_\cap+1,q_b-q_\cap+1}.
\end{align*}
It remains to justify the last inequality.
Let $n^* = \min(q_a-q_\cap + 1, q_b-q_\cap+1)$. Assume without loss of generality that $q_a-q_\cap + 1$ is smaller. Our last observation is that we can grant Alice and Bob the knowledge of a \emph{random} subset of $\supp(b)\setminus \supp(a)$ of size $q_b - q_a$. This implies that
\begin{align*}
    &~~~~ \adv_{1,q_a-q_\cap+1,q_b-q_\cap+1} \\
    &\le \adv_{1, q_a-q_\cap+1,q_a-q_\cap+1} \\
    &\le \epsilon^2 C\cdot O\left(\max\left(\frac{1}{q_a-q_{\cap}+1},\frac{1}{q_b-q_\cap+1}\right)\right),
\end{align*}
The last line follows because we have reduced the distinguishing game to a version of the Noisy Hidden Index Problem with generalization observations (c.f.~\Cref{def:nhip-general}), and the inequality follows by \Cref{lem:nhip-general}.
\end{proof}
Given the claim, we can make use of \Cref{lem:helper} (described at the end of the subsection) and deduce that
\[
\adv\leq O(1)\cdot \epsilon^2 C\cdot \E\left[\frac{q_\cap}{q_a-q_\cap+1}+\frac{q_\cap}{q_b-q_\cap+1}\right]\leq O(1)\cdot \frac{\epsilon^2 C\cdot \theta}{1-\theta},
\]
as desired.
\paragraph*{Proof for $\theta \in (1/2,1]$.} We now establish the proof for the case of $\theta > \frac{1}{2}$. We will in fact prove a stronger statement, which implies \Cref{thm:general-hidden-sign} for all $\theta \ge \frac{1}{2}$. Namely, we prove that, for $\theta = 1$, it holds that $\adv^{GHS}_{1,\epsilon,0} \le O(\epsilon^2 \sqrt{C})$.

First, by a reduction argument similar to the proof of \Cref{lem:nhip-general}, we can without loss of generality consider the case that $\Gamma_r \equiv \ber(\frac{1}{2} - \epsilon r)$ for $r\in \{\pm 1\}$. In this case, with $\theta = 1$, Alice and Bob are tasked to distinguish between the following two cases:
\begin{itemize}
\item Case 1: the inputs for Alice and Bob are a pair of independently generated bit strings $a\sim \{0,1\}^n, b\sim \{0,1\}^n$.
\item Case 2: the inputs for Alice and Bob are a pair of $\epsilon^2$ correlated bit strings: namely $a$ and $b$ and marginally uniform and every pair of bits $(a_i,b_i)$ have correlation $\epsilon^2$.
\end{itemize}
We will make use of a powerful result from \cite{HadarLPS19}. To describe the result, we begin with necessary notation and setup. Let $(X,Y)$ be the inputs to Alice and Bob in the correlated case, and $(\overline{X},\overline{Y})$ be the inputs in the uniform case. Let $\Pi$ be a communication protocol between Alice and Bob.

Let $P^1_{XY\Pi}$ be the distribution of $(X,Y,\Pi)$ (with the understanding the $\Pi$ is induced from the inputs $(X,Y)$), and similarly $\overline{P}^0_{XY\Pi}$ the distribution of $(\overline{X},\overline{Y},\Pi)$. Let $P^1_{X\Pi}$ be the marginal distribution of $P^1_{XY\Pi}$ on the $(X,\Pi)$ part. Define $P^0_{X\Pi}$ similarly. 

Now, in our language, a remarkable result of \cite{HadarLPS19} says that (see their Remark 3):
\begin{align}
D_{KL}(P^1_{X\Pi} \| P^0_{X\Pi}) \le \epsilon^4 I(\Pi;X,Y). \label{equ:HLPS-sdpi}
\end{align}
Let $\xi$ be the advantage of $\Pi$ on distinguishing $(X,Y)$ from $(\overline{X},\overline{Y})$. We set up the following experiment: flip a coin $R\sim \{1,2\}$. Depending on $R$ being $1$ or $2$, draw inputs $(x,y)$ from either $(X,Y)$ or $(\overline{X},\overline{Y})$. Finally run the protocol $\Pi$ on $(x,y)$ and obtain a verdict of $R$, denoted by $R'$. By the assumed advantage of $\Pi$, we know that $\Pr[R = R'] \ge \frac{1+\xi}{2}$. As a consequence, by Pinsker's inequality it follows that $I(R;R') \ge \frac{\xi^2}{4}$. Since we can obtain $R'$ from the communication protocol $\Pi$, it follows that
\begin{align*}
    \frac{\xi^2}{4} 
    &\le I(R;\Pi) \le I(R; \Pi, X) \\
    &\le \frac{1}{2} D_{KL}(P^1_{X\Pi} \| P^0_{X\Pi}) + \frac{1}{2} D_{KL}(P^0_{X\Pi}\| P^0_{X\Pi}) &\text{(the ``radius'' property of mutual info)} \\
    &\le \frac{1}{2}\epsilon^4 I(\Pi;X,Y) \le \frac{\epsilon^4}{2} H(\Pi) \le \frac{\epsilon^4 C}{2}. & \text{By \eqref{equ:HLPS-sdpi}}
\end{align*}
Re-arranging the inequality gives the desired upper bound on $\xi$, namely $\xi \le O(\epsilon^2\sqrt{C})$. 

Here we briefly explain the ``radius property'' of mutual information: for a joint distribution $P_{UV}$ with marginals $P_U$ and $P_V$, and for an arbitrary distribution $Q_V$, we have
\begin{align*}
    I(U;V) &= D_{KL}(P_{UV}\| P_U\times P_V) \\
    &=\mathbb{E}_{u\sim P_U}D_{KL}(P_{V\mid U = u} \| P_V) \\
    &\le \mathbb{E}_{u\sim P_U}D_{KL}(P_{V|U=u} \|Q_V).
\end{align*}
In our derivation, we used the property on the term $I(R; \Pi,X)$, using a reference distribution $Q_{X\Pi} = P^0_{X\Pi}$.
\end{proof}

\begin{lemma}\label{fact1}[Lemma 3 in \citep{ramdas2019unified}]
Let $K\sim\bin(t,\theta)$, then $\E[\frac{1}{K+1}]\leq \frac{1}{\theta(t+1)}$.
\end{lemma}
\begin{lemma}\label{lem:helper}
Let $A\sim\bin([n],\theta)$, $B\sim\bin([n],\theta)$. Then we have
\[
\E\left[\frac{|A\cap B|}{|A\setminus B|+1}\right]\leq \frac{\theta}{1-\theta}.
\]
\end{lemma}
\begin{proof}
Let $X_i:=[i\in A\cap B]$, $D_i:=[i\in A\setminus B]$. Let $X=\sum X_i$, $D=\sum_i D_i$. Then
\begin{align*}
\E\left[\frac{|A\cap B|}{|A\setminus B|+1}\right]&=\E\left[\frac{X}{D+1}\right]=n\cdot \E\left[\frac{X_1}{D+1}\right]=n \left(\Pr[X_1=1]\cdot\E\left[\frac{1}{Y+1}\mid X_1=1\right] \right)\\
&\leq n\cdot \theta^2\cdot \frac{1}{\theta(1-\theta)n}=\frac{\theta}{(1-\theta)}.
\end{align*}
Here, the last inequality is by the fact that $\Pr[X_1]=\theta^2$ and Lemma \ref{fact1}.
\end{proof}

\subsection{Concluding the Proof}

We are now ready to analyze the game from \Cref{def:d-game}. We state the following theorem.

\begin{theorem}
For every $\epsilon, m$, $n\leq O(m/\epsilon^2)$, and $C\geq \log m$, it holds that that any $C$-bit protocol for the Hidden Sign Problem achieves a distinguishing advantage of at most $O\left(\log^2(m)\cdot \frac{n}{m} \cdot C\cdot \epsilon^2 \right)$.
\end{theorem}

\begin{proof} 

Depending on whether $n\ge \frac{m}{10}$ or not, we consider two cases.

\paragraph*{Case 1: $n\ge \frac{m}{10}$.} 
    We first argue for the case of $n\ge \Omega(m)$. In this case, define $\Gamma_r$ for $r\in \{\pm 1\}$ to be the following distribution:
    \begin{itemize}
        \item Draw $t\sim \poi(n/m)$, and sample $t$ independent bits $x_1,\dots, x_t\sim \ber(\frac{1}{2}-\epsilon r)$. Output $(t, x_1,\dots, x_t)$.
    \end{itemize}
    Our primary claim is that $\Gamma_0$ and $\Gamma_1$ are $(\epsilon',\delta)$-indistinguishable for $\delta = \frac{1}{m^{10}}$ and $\epsilon' = O(\epsilon\sqrt{\frac{n}{m}}\log(1/\delta))$. To see this, first note that $\Pr[\poi(n/m) \le \frac{n}{m}\log(1/\delta)]\ge 1- \frac{1}{10\delta}$. We first condition on this event. Then, letting $t = O(\frac{n}{m}\log(1/\delta))$ and by \Cref{lem:dp}, we have that $\bin(t, \frac{1}{2}+\epsilon)$ and $\bin(t,\frac{1}{2}-\epsilon)$ are $(O(\epsilon\sqrt{t\log(1/\delta)}),\delta/2)$-indistinguishable. Combining these two pieces of observation, we conclude that $\Gamma_0$ and $\Gamma_1$ are $(\epsilon',\delta)$-indistinguishable.

    With this in mind, we can use \Cref{thm:general-hidden-sign} with $\Gamma_0,\Gamma_1,\theta = 1$ to deduce that the distinguishing advantage is at most $O(C{\epsilon'}^2)\le O(\log^2(m)\cdot \frac{n}{m}\cdot C\cdot \epsilon^2)$, as claimed.
    
\paragraph*{Case 2: $n\le \frac{m}{10}$.} We turn to analyze the case of $n\le \frac{m}{10}$. We first note that $\Pr[\poi(n/m)>0] = 1- e^{-n/m} \le \frac{n}{m}$. With this in mind, define $\Gamma_r$ for $r\in \{\pm 1\}$ as follows:
\begin{itemize}
\item Draw $t\sim \poi(n/m)|_{\poi(n/m)\ge 1}$ and sample $t$ independent bits $x_1,\dots, x_t\sim \ber(\frac{1}{2} - \epsilon r)$. Output $(t,x_1,\dots, x_t)$. 
\end{itemize}
In this case, we claim that $\Gamma_0$ and $\Gamma_1$ are $(\epsilon',\delta)$-indistinguishable with $\delta = \frac{1}{m^{10}}$ and $\epsilon' = O(\epsilon\log(1/\delta))$. To see this, we first observe that $\Pr[\poi(n/m)\le 5\log(1/\delta) \mid \poi(n/m) \ge 1] \le \delta / 10$. Let us condition on this event that $t\in [1,5\log(1/\delta)]$. Then, similarly as the above reasoning, we use the fact that $\bin(t,\frac{1}{2}-\epsilon)$ and $\bin(t,\frac{1}{2} + \epsilon)$ are $(O(\epsilon\sqrt{t\log(1/\delta)}), \delta/2)$-indistinguishable to deduce that $\Gamma_0,\Gamma_1$ are $(\epsilon',\delta)$-indistinguishable.

Finally, we can make use of \Cref{thm:general-hidden-sign} to deduce that the distinguishing advantage in this case is upper-bounded by $O(C{\epsilon'}^2\theta) \le O(\log^2(m)\frac{n}{m}C\epsilon^2)$, as claimed.
\end{proof}

\section{Conclusion}\label{sec:conclusion}
We established a multi-pass streaming lower bound for uniformity testing over a domain of size $[2m]$: any $\ell$-pass algorithm that uses $n$ samples and $s$ bits of memory must satisfy $\ell sn=\tilde\Omega(m/\epsilon^2)$. This extends the unconditional one-pass bound of \cite{colt19} to the multi-pass setting. A remaining question is whether the $\epsilon$-dependence can be improved to match the upper bound $sn=\tilde{O}(m/\epsilon^4)$.

\bibliographystyle{alpha}
\bibliography{main.bib}


\end{document}